\documentclass{article}
\usepackage{hyperref}

\usepackage{latexsym}
\usepackage{calc}
\usepackage{url}
\usepackage{rotating}
\usepackage{enumerate}
\usepackage{amsmath} 
\usepackage{amsthm} %%% Theorem environments with `amsthm' %%%
\usepackage{ifthen}
\usepackage{lineno}

\usepackage{float}

\usepackage{graphics}
\usepackage{graphicx}
\usepackage{color}
\usepackage{wrapfig}

\usepackage[noend]{algorithmic}
\usepackage[boxed]{algorithm}

\usepackage{tocloft}
\floatname{algorithm}{Figure}

\definecolor{DarkBlueLinks}{RGB}{5,32,144} % more blue, less green

\newtheorem{theorem}{Theorem}

\newtheorem{lemma}[theorem]{Lemma}

\theoremstyle{definition}
\newtheorem{definition}{Definition}
\theoremstyle{assumption}

\theoremstyle{remark}

\renewenvironment{proof}{\medbreak\noindent{\em Proof.}}{\ $\Box$\medbreak}
\newenvironment{proof-sketch}{\medbreak\noindent{\em Proof (sketch).}}{\ $\Box$\medbreak}

\newenvironment{revpar}{%
\begin{list}%
{}{\setlength{\labelwidth}{0in}%
\setlength{\itemsep}{0in}%
\setlength{\topsep}{0in}%
\setlength{\leftmargin}{3em}%
\setlength{\itemindent}{-2em}}}{\end{list}}
\begin{document}
\thispagestyle{empty} %\cleardoublepage

\title{Wendy, the Good Little Fairness Widget}
\author{Klaus Kursawe (klaus@vega.xyz)\\ Version 1.01}

\maketitle
\pagestyle{plain} 
%\pagebreak

\begin{abstract}
The advent of decentralized trading markets introduces a number of new challenges for consensus protocols. In addition to the `usual' attacks -- a subset of the validators trying to prevent disagreement --
there is now the possibility of financial fraud, which can abuse properties not normally considered critical in consensus protocols. We investigate the issues of attackers manipulating or exploiting the order in 
which transactions are scheduled in the blockchain. More concretely, we look into {\em relative order fairness}, i.e., ways we can assure that the relative order of transactions is fair.  We show that one of the more
intuitive definitions of fairness is impossible to achieve. We then present Wendy, a group of low overhead  protocols that can implement different concepts of fairness. Wendy acts as an aditional widget for an existing blockchain, and is largely agnostic to the underlying blockchain and its security assumptions. Furthermore, it is possible to apply a the protocol only for a subset of the transactions, and thus run several independent fair markets on the same chain. 

\end{abstract}

%\pagestyle{fancyplain}
%\pagenumbering{arabic}
%\rhead{\leftmark}
%\lhead{Relative Fairness in Blockchains}
%\cfoot{\thepage}
%\lfoot{}

\section{Introduction}
\label{Introduction}

%In terms of protocol performance, one of the problems is that it is quite hard to perform a meaningful comparison between different %consensus protocols - while many metrics exist~\cite{Bano2017SoKCI}, it is hard to gauche the real world impact, and sometimes optimizing %for one metric sacrifices on another one. Furthermore, some parameters - such as network behaviour or the number of validators - can %change quite significantly over the lifetime of the protocol. We adress this issue by int
The advent of decentralized trading markets introduces a number of new challenges for consensus protocols~\cite{DBLP:journals/corr/abs-1904-05234,VegaWhitepaper}. Classically, consensus layer protocols only are required to maintain consistency of
the blockchain. While additional requirements have been investigated in the past -- for example causal order or censorship resilience -- very little attention has been given to the fairness
of the order of events, making it possible to execute frontrunning or rushing attacks. While some blockchains attempt to make such attacks harder, for example by using a randomized leader election protocol, 
others can be easily manipulated by a single corrupt validator or a well targeted denial of service attack. In addition to allowing questionable behavior, this can also be a potential regulatory issue, if
exchange are required to prevent some levels of fraud.

In this paper,  we investigate the issues of attackers manipulating or exploiting the order in 
which transactions are scheduled in the blockchain. More concretely, we look int {\em relative order fairness}, i.e., ways we can assure that the relative order of transactions is fair.  We show that one of the more
intuitive definitions of fairness is impossible to achieve, and present several alternatives.

Our approach integrates with existing blockchains without any change or non-standard assumption on the blockchain implementation -- the only requirement is that there is some known set of parties (resp. validators) through 
which fairness is defined. This allows us to combine several variations of fairness with different blockchains, have different degrees of fairness for sets of transactions running on the same blockchain, and even change the configuration
on the fly without needing to break the chain.
This setup can also  come especially handy if one wants to formally verify the protocols -- it is vital here to have the small, independently verifiable components and to not need to formally verify dozens of variations of the same protocol (a glimpse at the difficulty of formally verifying consensus protocols can be found in
~\cite{Kwiatkowska02verifyingrandomized}).

\section{Model and Architecture}
We assume a two-fold model. For one, there is an underlying blockchain that takes blocks as an input and produces a distributed ledger of blocks. For the purpose of the fairness
widget, we do not require any assumptions on the blockchain regarding participants, timing, or finality. What we do require is that the blockchain has some form of validity
function that evaluates if a block is valid, and that can include the validity conditiond for the fairness widget. We also assume that all validators that can propose new
blocks for the blockchain that include fairness-relevant transactions are known and can receive a broadcast from the validators participating in the pre-protocol.

There is no requirement that all blocks in the blockchain are subjected to realtive fairness. If, for example, Ethereum was the underlying blockchain, the fairness protocol
could be required for all transactions touching some specific smart contract, without putting any requirements onto other transactions. To this end, all transactions
that reuqire fairness contain a fairness-label, and only transactions with the same label need to be fair with respect to each other. Similarily, not all validators need to 
participate - it is possible that only a subset of the validators propose fairness relevant blocks, though this would slow down the fair transactions.

The fairness preprotocol itselfs requires a more strict model. 
Our model extends the system model and definitions of Cachin, Kursawe, Petzold, and Shoup~\cite{CKPS01}. Thus, we assume that the number of byzantine corrupted parties is just less than a third of all parties (i.e., $n=3t+1$), though we will also show how to expand the approach to a more flexible model~\cite{Kursawe05byzantinefault}.  These could be a subset
of the validators of the underlying blockchain, or a completely independent set of parties. We work in the fully asynchronous model, i.e., we assume that an attacker has complete control over the time and order of message delivery,
 but is not allowed to completely drop a message. Furthermore, we assume that messages are authenticated, and that all participants can sign messages as well as verify each other's signatures.
In addition to the classical byzantine nodes,  we also assume
rogue traders might try to game the system to get an unfair advantage, especially to get ahead on 
performing a transaction. These traders can collaborate with any amount of other traders as well up to a third of the validators; in fact, formally we assume that all traders are
under control by the adversary.

The validators receive external {\em requests} from the traders. We make no assumption
on the timing or the order, which is under complete control from the adversary. The blockchain protocol then {\em delivers} the requests, i.e., puts them into a block while satisfying
the basic properties of atomic broadcast. In practice, to optimize bandwidth, the protocol would likely not use the requests themselves, but hashes thereof. For the sake of presentation,
we will use the term request even when a hash would be sufficient. Messages are send by a simple multicast with no requirements on consistency or safety. While there might be some
room for optimization if intelligent gossiping protocols are used, our only requirement for the communication layer is that messages between honest parties eventually arrive.
An alternative model in the literature is the GST model going back to Dwork, Lynch, and Stockmeier~\cite{10.1145/42282.42283}, which in some interpretations does allow for some message loss.
 In this model, the adversary is allowed to arbitrarily delay or drop messages until a time
called the {\em global stabilization time}, after which she needs to deliver all messages within a known timeout. In this model, protocols essentially try to not violate safety before GST, 
and then assure liveness after. While we don't model our protocols in this setting, they work in it just as well as long as lost messages are resend.

As mentioned above, the goal of our design is not to build a new blockchain that includes fairness, but to build a module that can be added to existing blockchain designs.
 To this end, we provide
a pre-protocol that is run by the validators in parallel to the actual blockchain. The pre-protocol outputs valid blocks that assure relative order fairness. While these blocks can
be generated by every validator,  in most consensus  implementations, blocks are proposed by only one or very few parties. To this end, we define a set of designated leader(s) which
execute the part of the protocol that generates blocks. The leader part does not involve any communication though, and thus could be executed by every
participant without additional communication effort. In addition, we need to modify the block-validity function -- proposed blocks are not valid unless it is also verified
that the fairness conditions have been satisfied. 
To be able to use more established formal definitions, we assume that our protocol communicates with an atomic broadcast subprotocol; for all practical purposes, this
is equivalent to a blockchain in our context. We make no assumptions on how the underlying atomic broadcast protocol is implemented,
and what -- if any -- timing assumptions it uses. In fact, our preprotocol can work in a completely different model than the underlying blockchain -- while our
model has a voting/quorum based approach in mind, the blocks generated by the fairness pre-protocol can es well be processed by a Nakamoto style implementation
such as Ethereum or Ouroboros, not unlike the approach that Casper is taking to add finality~\cite{DBLP:journals/corr/abs-1710-09437}.
 We do, however, assume that the participants in the fairness protocol know and recognize each other. While it
thus would be logical to assume the same for the blockchain protocol, this is not strictly necessary -- it is  possible to use the fairness protocols presented here
to add relative order fairness to (some) bitcoin transactions, as long as it is possible to enforce our new validity condition for that chain and assure the the underlying blockchain only accepts
the blocks we generate in the order we generated them.

As we envision a blockchain that handles a diversity of transactions, relative order fairness only needs to be assured for subsets - it is not necessarily required that a request related to
technology a stock market is treated fairly with respect to a request related to crop prices in Australia. Thus, every transaction has a market-identifier $m_{id}$, and only transactions
that have the same market-identifier need to be fair with respect to each other. As we provide different fairness models, it is also possible to use different fairness pre-protocols for
different markets. There is even a possibility that a single request has several market identifiers and thus is delivered in a relatively fair way with respect to several, otherwise independent
markets. The main issue with this model is that it adds quite some complexity if we want to have different fairness protocols for different markets. While there is no fundamental issue with
this, we do not include this property for our protocols in this paper for the sake of (relative) simplicity.

\subsection{Related Work}
The only work we are aware of that looks at relative fairness is parallel work from Kelkhar, Zhang, Goldfeder, and Jules~\cite{cryptoeprint:2020:269}. They also identify the impossibility of 
strict fairness  and resolve to address block fairness. While our approach is to weaken the fairness condition to circumvent the impossibility of block fairness, they define a concept of
weak liveness wile maintaining the stronger fairness condition to this end, and define a set of protocols (both synchronous and asynchronous) to provide relative- or order block fairness. The price for the stronger fairness
is that there is no limit on when requests are delivered or how big a block becomes, though the protocols could easily be adapted  to one of our models. Their approach also differs in the architecture - 
while we aim to have a module to be combined with existing atomic broadcast protocols, their work presents a full protocol for $ n > 4t$. 

%\subsubsection{Causality and Commit \& Reveal}
The concept of {\em causality} in state machine replication was first introduced by Birman and Reiter~\cite{10.1145/177492.177745}, with the example of preventing 
stock trading fraud. The definition was later refined by Cachin, Kursawe, Petzold, and Shoup~\cite{CKPS01}, and again by Duan, Reiter, and Zhang~\cite{8023111}.
%[39] M. K. Reiter and K. P. Birman, “How to securely replicate services,” ACM Traovertake a particula request is limited.nsactions on Programming Languages and Systems, vol. 16, pp. 986–1009, May 1994.
%Causality: If deliver(m') is executed at a correct or honest server s when deliuer( m) has not yet been executed at s, then m’ was issued before m was decrypted anywhere (other than c).
While the details in the definitions do matter for meaningful proofs and avoiding less straightforward attacks, the basic idea of these definitions is the same; 
a message is processed by the protocols in a way that its position in the ordering is fixed before any participant learns of its content. 
While this is sufficient to prevent some financial fraud -- especially if we also allow the sender of a request to remain anonymous until the transaction is scheduled -- the
protection offered by commit and reveal is not sufficient. Especially in cases of high volatility, traders can still get an advantage if they can schedule transactions faster than their peers.   

The notion of {\em fairness} has been used in different contexts in the literature. In the context of block delivery, the concept was formally introduced in~\cite{CKPS01}, though some
extend of fairness is already provided by earlier protocols such as Castro and Liskovs BFT protocol~\cite{DBLP:conf/osdi/CastroL99}. In this definition, fairness essentially requires that a 
blockchain is fair if the time between $t+1$ honest parties being aware of a request and that request being delivered is bounded. 
This concept is somewhat similar (and sometimes used as a synonym) to {\em censorship resilience}~\cite{Miller:2016:HBB:2976749.2978399}, though that term as well has now taken on a multitude of meanings in the literature, 
and usually does not rule out an unfair delay in delivering a request. In terms of {relative order fairness}, fair protocols at least give an upper bound on the level of unfairness -- while it is
possible that requests are processed in a different order than they arrived, the number of requests that can rush ahead of a particular request is limited.
In~\cite{DBLP:journals/corr/abs-1906-03819}, a different fairness definition is defined -- here, fairness requires that
all validators get an equal opportunity to get their transactions into the blockchain. This is a different model than we assume, as we want to achieve fairness for transactrions comming from external participants, while this protocol
assures fairness between the validators. There is some relation though, as fairness between validators assures that
the dishonest validators cannot dominate the blockchain, and thus requests seen by all honest validators
are processed somewhat fast.
%XXXThere are a number of subleties in marrying the notion of limited time with asynchronous adversaries, but this is not relevant for our argument at this point.

%Measuring time does not make much sense in the asynchronous model, this fairness definition measures time through the number of implementation messages send by honest parties.
%In the timing based leader protocols based on Castro-Liskov~\cite{CL99}, the definition is more complex - there is usually no limit on how many leaders can be disposed of due to a timeout, and thus no upper bound
%on both time and effort that needs to be spent to deliver a request. With the exception of Kursawe-Shoup\cite{KS01}, which has a fully asynchronous revocery phase in case of
%a leader failure, these protocols can thus only guarantee fairness once some network timing assumption holds, e.g., in the GST (global stabilisation time) model. In addition, it is possible
%to make assurances about an unfair leader - if a leader tries to omitt a particular request too long, it can be replaced by another leader.

The proof-of-work model has a different approach to fairness. Essentially, if the majority of miners are honest, and the number of transactions is smaller than the maximum the network can handle, 
the probability that some winning miner will process a given transaction soon is relatively high (though there is no strict upper bound). This effect is diluted by an economic argument though -- 
if (as the case in Ethereum and Bitcoin) it is possible to pay miners for preferred treatment, the delay until a particular request is delivered can become fairly high. In terms of relative fairness,
this feature makes the blockchains unfair by design -- it is explicitly build in that clients who pay more can get preferred treatment.

%Fair BFT protocols For one of our constructions, fairness plays a vital role. Fairness prevents the BFT service from unfairly delaying or dropping some clients’ requests but not others. For instance, Aardvark \cite{27} implements a fair and %efficient sequencer-based BFT protocol. (The idea is that if the sequencer fails to enforce fairness, a view change will be triggered.) ByzID \cite{34} is another fair BFT protocol (that uses small trusted components).
%üDaun reiter yhang

%Ethereum and Bitcoin : Capitalistic fairness

%. Input causality is related to the standard causal order (going back to Lamport [24]), which is a useful safety property for distributed systems with crash failures, but is actually not well defined in the Byzantine model [21]. 

Some of the more recent protocols~\cite{DBLP:journals/corr/abs-1807-04938,DBLP:journals/corr/abs-1803-05069} frequently exchange the leader even in the absence of observable misbehavior. This makes it harder for an attacker to impose controlled unfairness, as it is harder to assure a corrupted validator is in charge of scheduling when the adversary needs it, though it might be possible to remove the honest leader with a limited denial of
service attack. An additional countermeasure is to choose the next leader randomly, decreasing another level of control of the adversary. Fully randomized protocols~\cite{CKPS01,Miller:2016:HBB:2976749.2978399} also make it
harder for an attacker to control the level of unfairness. Nevertheless, an attacker can still cause unfairness to a large extent, and -- while the unfairness is harder to control -- the protocols are not necessarily {\em relative fair}, 
i.e., preserving the order in which requests come when delivering them.

\section{Relative Fairness}

The term fairness has found numerous definitions in the atomic broadcast and blockchain literature. The most relevant definition
of (absolute) fairness for our context requires one of the following:
\begin{itemize}
 \item every request eventually gets scheduled
 \item every request gets scheduled within a bounded time or number of implementation related messages
\end {itemize}
Additional constraints depend on the model used, e.g., requests only need to be scheduled within a bounded time after GST (Global Stabilisation Time)~\cite{10.1145/42282.42283}.

For many consensus protocols, satisfying this definition of fairness does not come naturally. Especially for leader-based protocols, a leader can easily suppress a message. There are a number of countermeasures against this. In~\cite{DBLP:conf/osdi/CastroL99, DBLP:conf/icalp/KursaweS05}, replicas watch a leader and dispose of them if they are dishonest; other protocols\cite{DBLP:journals/corr/abs-1803-05069,DBLP:conf/netys/Amoussou-Guenou19} change the leader frequently, in the hope that an honest leader will eventually handle all outstanding requests. With the exception of~\cite{DBLP:conf/icalp/KursaweS05}, no protocol can give strong bounds on when a message is actually scheduled -- the time until a message gets scheduled depends on the accuracy of the timing assumptions (or the arrival of GST) and is thus dependent on an out-of-protocol factor.
%XXX Too strong
Leaderless protocols~\cite{CKPS01,Miller:2016:HBB:2976749.2978399} tend to have better implicit fairness protection; while they tend be a little slower than leader based ones (at least in a well-behaved network), the decreased effort to assure fairness can give those protocols an edge in a trading blockchain.

As we are anyhow sorting transactions into blocks (this comes rather natural for a blockchain), though it is possible to
use logical blocks that encompass several blockchain blocks. In addition to relative fairness, this also assures fairness as 
defined above. The pre-protocol each party would follow looks as follows (unoptimized version, basing on a leader based
atomic broadcast protocol for simplicity):

\begin{definition}[Block Fairness]
After a request has been seen by $n-t$ honest parties, it will be scheduled in the next block; if it hasn't been seen by at least one honest party, it will not be scheduled in the next block.
\end{definition}

This is relatively easy to implement -- before the ordering protocol starts, every validator sends  around a list of all requests they have seen; a valid proposal for a block then consists of the transactions out of $n-t$ of these sets that got $t+1$ votes. 

In the setting we envision for our blockchain, even the stronger definition of fairness is insufficient. In addition to the requirements of absolute fainress, we also want {\em relative fairness}, which more captures the intuitive meaning of the word -- if one request is send before another request, it would be fair if it is also scheduled first.

\begin{definition}[Relative Fairness]
A byzantine fault tolerant total ordering protocol is called {\em relatively fair} if the following holds: If all honest parties receive request $r_1$ before request $r_2$, then $r_1$ is delivered before $r_2$.
\end{definition}
Unfortunately, we can show that this definition of fairness is not only impossible, but inherently contradictory even if only one party is corrupt.

\begin{proof-sketch}
Suppose we have $n$ parties $P_1$, ..., $P_n$, and $n$ requests $r_1$, ..., $r_n$.
Then let $P_i$ get the requests in the order $r_i,r_{i+1},r_{i+2}$ $,..., r_n, r_1,r_2,...,r_{i-1}$.
Now for every $j$, the only party that sees $r_j$ before $r_{j-1}$ is party $P_j$; all other parties
see $Pr_{j-1}$ before $r_j$; also, $P_1$ is the only party that sees $r_1$ before $r_n$.

If all parties are honest, then there is no dedicated message order -- no two requests will have been seen
in the same order by all honest parties. However, if any party $j$ is dishonest, then $r_j$ must be scheduled
after $r_{j-1}$, as $P_j$ is the only party to see $r_j$ before $r_{j-1}$ (if $P_1$ is dishonest, $r_n$ must be scheduled before $r_1$). 

As the honest parties following the protocol do not know who is dishonest, the outcome of the 
ordering protocol must be correct independently of which party is dishonest. Thus, for all $i$,$r_i$ must be scheduled before $r_{i+1}$ as well as $r_n$ before $r_1$, which is a contradiction.
\end{proof-sketch}

One way out would be to only require $r_2$ and $r_1$ to be in the same block. However, even that might not be possible, and there is another weakness in this definition: The corrupt parties might see $r_2$ long before any honest party would see $r_1$, thus our protocol essentially can't schedule anything seen by $t$ parties only; it seems hardly fair if $t$ validators cannot get a message scheduled that every client can schedule.  We leave it to further work to find further definitions for relative fairness that are efficiently achievable and might serve some usecases better.

\begin{definition}[Relative Fairness, 2. attempt]
A byzantine fault tolerant total ordering protocol is called {\em relatively fair} if the following holds: If all honest parties receive request 
$r_1$  before request $r_2$, then $r_1$ is delivered in the same block as $r_2$ or earlier.
\end{definition}

Unfortunately, we can show that this is also impossible:

\begin{proof-sketch}
In above proof, we have shown that there exists a schedule in which the required order of messages depends on which party is
faulty, thus requiring to take into account a parameter that is not known to an honest party. In this proof, we build on that construct
to design a schedule that would create a block of unlimited size.

For this outline, we assume $n=4$ and $t=1$. Consider two schedules as used above, i.e., 
~\\
$P_1$: $m_1$, $m_2$, $m_3$,$m_4$\\
$P_2$: $m_2$, $m_3$, $m_4$,$m_1$\\
$P_3$: $m_3$, $m_4$, $m_1$,$m_2$\\
$P_4$: $m_4$, $m_1$, $m_2$,$m_3$\\

and

~\\
$P_4$: $m_5$, $m_6$, $m_7$,$m_8$\\
$P_3$: $m_6$, $m_7$, $m_8$,$m_5$\\
$P_2$: $m_7$, $m_8$, $m_5$,$m_6$\\
$P_1$: $m_8$, $m_5$, $m_6$,$m_7$\\

Both schedules area split into three segments as shown below:

\begin{table}[H]
\centering
\begin{tabular}{l|lll|llll|l}
     & $A_1$ &          &          & $A_2$ &          &          &          & $A_3$  \\ 
\hline
$P_1$ & $m_1$ & $m_2$ & $m_3$ &          &          &          &          & $m_4$  \\
$P_2$ &  & $m_2$ & $m_3$         & $m_4$ & $m_1$ &          &          &           \\
$P_3$ &  &          &    $m_3$      & $m_4$ & $m_1$ &   $m_2$       &          &           \\
$P_4$ &          &          &          & $m_4$ & $m_1$ & $m_2$ & $m_3$ &          
\end{tabular}
\end{table}

\begin{table}[H]
\centering
\begin{tabular}{l|lll|llll|l}
     & $B_1$ &          &          & $B_2$ &          &          &          & $B_3$  \\ 
\hline
$P_1$ &            &  &  &  $m_8$        & $m_5$         & $m_6$         &  $m_7$        &  \\
$P_2$ &            &  & $m_7$ & $m_8$ &   $m_5$       & $m_6$         &           \\
$P_3$ &   & $m_6$      & $m_7$         & $m_8$ & $m_5$ &          &          &           \\
$P_4$ & $m_5$ & $m_6$      & $m_7$  & & &  & &     $m_8$     
\end{tabular}
\end{table}

We now link those two schedules to one combined schedule with the segment order $A_1$, $B_1$, $A_2$, $B_2$, $A_3$, $B_3$.

By the design of schedules A and B, to achieve fairness, $m_1$, $m_2$, $m_3$, and $m_4$ must be in the same block. The same holds for $m_5$, $m_6$, $m_7$, and $m_8$. 
The argument for this is equivalent to the previous proof; as it is not known to the honest parties who is honest and who not, the requirement could
imply that $m_1$ has been seen by all honest parties before $m_2$ (if $P_4$ is corrupt), $m_2$ before $m_3$, $m_3$ before $m_4$, and $m_4$ before $m_1$.
Thus, all those messages need to be scheduled in the same block.

In the combined schedule, we also have all honest parties see $m_7$ before $m_4$. Thus, $m_7$ must be scheduled in the same or an earlier block than $m_4$.
Similarly, $m_3$ needs to be in the same or an earlier block than $m_8$.  As $m_7$ and $m_8$ and respectively $m_3$ and $m_4$ must be in the same block, this means all messages have to be scheduled in the same block.

If we combine the segments the other way around, i.e., $B_1$, $A_1$, $B_2$, $A_2$, $B_3$, $A_3$, we get the same result: $m_7$ is seen by all parties before $m_4$, and $m_3$ is seen
by all parties before $m_8$, meaning that still both segments need to be in the same block.

We can now repeat this construction. Suppose we have segment $C$ in the same structure as segment $A$, and segment $D$ in the same structure as segment $B$. Then consider the schedule

$A_1$, $B_1$, $A_2$ ,$C_1$, $B_2$, $A_3$, $C_2$, $B_3$, $C_3$

By above argument, all messages in $A$ and $B$ need to be in the same block; the addition of the messages from segment $C$ does not affect the argument. Similarily, all messages in $B$ and $C$ need to be
in the same block; this is unaffected by $A$. In the same way, we can add $D$ in a way that it needs to be in the same segment as $C$:

$A_1$, $B_1$, $A_2$ ,$C_1$, $B_2$, $A_3$, $D_1$,  $C_2$, $B_3$, $D_2$,$C_3$,$D_3$

This construction can be arbitrarily repeated, leading to an infinite sequence of messages that all need to be in the same block.
\end{proof-sketch}

A notable property of our result is that we do not need a corrupted party to actually act in any bad way -- it is enough that there is some party that has the label 'corrupt', and noone knows
which one it is.  While we haven't worked out the proof, it is likely even impossible if we only require fairness if noone actually is corrupt. To assure liveness in an asynchronous system, 
the protocol still needs to progress on $n-t$ inputs, which means it misses some information that might be relevant to define a valid order.  We did at this point not investigate
further, as we prefer to have a protocol that offers somewhat weaker fairness, but maintains robustness in the face of a byzantine adversary.

There are subtle differences in the underlying model that impact what the construction actually means. In some models -- essentially the cryptographically sound ones that assume
a polynomial time bound adversary~\cite{CKPS01,cryptoeprint:2020:269}-  one assumes that the number of incoming (and adversary generated) requests is somehow bounded, i.e., at some point the protocol
terminates for good. In this model, our construction does not strictly violate liveness -- what happens is that, to satisfy fairness, all requests will be delivered in the one and only
block the protocol ever schedules just prior to termination. For those models, we do not prove impossibility of relative block fairness, but impossibility of any meaningful
efficiency guarantees -- in the worst case, relative fairness is reached by treating all parties equally bad.
If we assume a model that allows for infinite protocol runs, the last point in time does not exist, and a protocol cannot guarantee to deliver anything.

The other interesting modeling aspect is the amount of asynchrony required. In the schedule above, 
once we start interleaving the D-blocks, all messages in the A-block have been seen by all honest parties. 
This implies that we do not need a fully asynchronous system. For a consensus between $n$ parties, 
if $\delta t_r$ is the time interval between the first honest party becoming aware of a request $r$ and
the last honest party doing so, then the adversary needs to show honest parties less than $3n$ other requests during $\delta t_r$.
Thus, our construction is also possible in most synchronous systems, as long as the adversary can generate/access sufficient requests in the
given time-span and has the power to freely determine a schedule in which an honest party sees any set of $3n$ consecutive requests.

%XXXXXX
Thus,  if we bound the number of requests the adversary is allowed to show to honest parties in between the
times when the first honest party saw a particular request 
%is a bounded number of messages the adversary can insert in the time between the first 
and the last honest party saw it, the impossibility result still holds.
%Thus, more correctly, what we have shown is
%XXXXXX

\begin{theorem}
There exists a schedule such that, to achieve relative block fairness,  all requests any honest party ever seen need to be scheduled in the same block. Consequently,
no block can be delivered with this schedule while new requests can be generated.

Furthermore, once an honest party has seen a request r, the schedule requires  less than $3n$ other requests to be operated on until the last honest party sees $r$.
Thus, an infinite schedule can also be generated in a partialy synchronous model.
\end{theorem}

\section{Circumventing the impossibility}

We first show a protocol that can guarantee fairness, but does not overcome the liveness issues mentioned above, i.e.,
it is possible for an adversary to prevent termination. For the ease of description, we describe a somewhat wasteful version of the 
protocol which resends all requests that did not make it into a block for the next block; in a real implementation, this would be handled 
in a more efficient way. Also, the protocol as described is sending a lot of signatures repeatedly; that, too, can be optimized in an 
implementation version.

We describe our protocol as a pre-protocol  to the atomic broadcast. The pre-protocol generates a proposal for a block that can then be proposed as the next block for the atomic 
broadcast protocol, alongside validation information that allow verifying that the block was properly generated.
To this end, we assume an atomic broadcast protocol following the definition of~\cite{CKPS01}. In addition to needing an external validity property, i.e.,  
there is a validation function such that an honest party only accepts an output $r$ with added validation information if the verification function holds.
be one party, or every party intending to construct a valid proposal. For simplicity, we also assume that the protocol is re-invoked upon termination by the 
atomic broadcast protocol, and that the framework assures that messages linked to undelivered requests are replayed to the next incarnation of the pre-protocol
in the same order, and messages linked to delivered requests are ignored. The reason to structure the protocol this way (rather than having an infinite loop that
invokes the atomic broadcast protocol and taking care of messages itself) lies in the modular architecture we want to allow - the fairness pre-protocol is an optional 
add-on to the atomic broadcast, and thus should be a pre-protocol invoked by the atomic broadcast rather than the other way around, and it must be possible for
one atomic broadcast protocol to use different pre-protocols for different markets.

One issue with this approach is that fairness in the traditional sense -- if every instance of the pre-protocol terminates, then every request that is seen by all honest party also is delivered (preferably in a bounded time) in some block -- is no longer a property of the pre-protocol, but of the combination. This can however easily be derived from relative fairness  
if we show that every terminated instance of the pre-protocol delivers a non-empty block:
\begin{itemize}
\item By assumption, messages that have not been delivered are treated by the next incarnation of the pre-protocol as if they arrived at the same time in the same order
\item The protocol guarantees progress, i.e., at least one request is delivered into a block on each terminating  incarnation
\item By the relative fairness requirement, for every request that has been seen by all honest parties, there is a finite number of requests that can be scheduled
        in an earlier block.
\end{itemize}.

%block: If the sum of parties that voted $r'$ before $r$ and those that haven't voted $r$ yet is at least $n-t$.
%We define a request $r$ to {\em block} another request $r'$ if it is not the case that $t+1$ parties reported having seen $r$ before $r$.

We say that a request $r$  {\em blocks} another request $r'$ given the current information, it cannot be excluded that $r$ needs to be in the same or an earlier block to achieve
relative block fairness. More precisely, $r$ {\em blocks} $r'$ if $r$ and $r'$ share a market-identifier, and it is not the case that $t+1$ parties
\begin{itemize}
\item have reported to have seen $r$ before $r'$, i.e., assigned it a lower sequence number, or 
\item have reported to have seen $r$ and all requests with a lower sequence number, but not $r'$.
\end{itemize}

\begin{lemma}
If $r'$ does not block $r$, then $r'$ is not required to be in the same or an earlier block than $r$ by the requirements of relative block fairness.
\end{lemma} 
\begin{proof}
To be required to be in the same or an earlier block, all honest parties need to have seen $r'$ before $r$. If $t+1$ parties report to have seen
$r'$ after $r$, at least one of them is honest, and thus not all honest parties have seen $r'$ before $r$.
\end{proof}

\begin{algorithm}\label{prot:weakfairness}
{\bf Widget Neverending Wendy  for block $b$ and protocol instance $ID$}\\
{\bf All parties:}
\begin{small}
\begin{revpar}
\item {\bf let} $i$ be the counter of incoming requests, starting at 0.
\item {\bf while} no valid proposal has been seen as the proposal for atomic broadcast for block $b$  {\bf do}
\begin{revpar}
 \item {\bf for all} known and unscheduled request $r$, in the order of the receiving the requests, send the signed message ($ID$,$b$,$i$,$r$) to all parties, where $i$ is the sequence number of that request.
 %\item {\bf if} a new request $\hat r$ arrives, send the signed message ($ID$,$b$,$i$,$\hat r$) to all parties
\end{revpar}
\item {\bf end while}

\end{revpar}
\end{small}
{\bf Additional protocol for the leader(s):}
\begin{small}
\begin{revpar}
\item {$\mathcal{B} = \emptyset$}
\item {\bf wait until } the first request $r$ is contained in the signed and valid votes from $n-t$ parties; add $r$ to $\mathcal{B}$
%\item {\bf while} any request $r'$ has at least $t+1$ votes with lower sequence numbers than the corresponding votes from any $r \in \mathcal{B}$,
\item {\bf while} any request $r' \not \in \mathcal{B}$ blocks any other request $r \in \mathcal{B}$,
\begin{revpar}
 %\item 	{\bf wait for} request $r'$  to get $n-t$ votes; add $r'$ to $\mathcal{B}$
\item 	{\bf if} request $r'$  has at least $n-t$ votes, add $r'$ to $\mathcal{B}$
\end{revpar}
\item {\bf end while}

\item The proposal for the next block of the atomic broadcast is $\mathcal{B}$, validated by all signed votes for requests in $\mathcal{B}$.
\end{revpar}
\end{small}
\end{algorithm}

The following defines how a valid vote and block look like:
\begin{definition}[Vote-Validity]
A vote is valid if it has the proper format, and once all requests with a lower sequence number from that voter have been received.
\end{definition}

\begin{definition}[Block-Validity]
A block $\mathcal{B}$ is valid if it contains a nonempty  set of requests with $n-t$ valid votes each; a vote for $r$ is valid if it contains the signed votes for all requests for that block with a lower sequence number. Furthermore, for every $r$ in $\mathcal{B}$, if there is a request $r'$ in the vote validation that had at least $t+1$ votes with a lower sequence number than $r$, then $r'$ needs to be in $\mathcal{B}$ accompanied 
by $n-t$ validation votes.
\end{definition}

\begin{theorem}
The protocol {\em Neverending  Fairness} guarantees safety, i.e., if a block is sent to the  atomic broadcast protocol, and there are requests $r$ and $r'$ such 
that all honest parties have seen $r'$ before $r$,  then $r'$ is in the same or an earlier block than $r$.
\end{theorem}

%\begin{lemma}
%Any $r \in R$ has been voted for by at least $n-t$ validators.
%\end{lemma}

\begin{proof}
%Let us first assume an honest leader. 
%If all honest parties saw $r'$ before any $r \in \mathcal{B}$, then $r'$ has at least $t+1$ votes that preceed the corresponding votes for $r$.
%Thus, the leader will enter the wait statement waiting for $n-t$ voters on $r'$.
%As every request any honest party sees will be seen by all other honest parties, eventually there will be $n-t$ votes and the wait will terminate with $r'$ being added to $B$.
If the leader is honest, it will place at least one request in $\mathcal B$. By the protocol logic, $\mathcal B$ will be delivered once no request not in $\mathcal B$ blocks any request in $\mathcal B$.

%If all honest parties saw $r'$ before $r$, then of any set of $n-t$ votes, $t+1$ will claim they saw $r'$ before $r$. Thus, if $r$ is in the block, and $r'$ is not yet scheduled, then $r'$ will be in the same block.
As the validity proof contains all the history that lead to the definition of the block, every valid block has to satisfy the conditions for relative block fairness.
If the leader is dishonest, the only misbehavior (apart from deliberately not terminating the pre-protocol) is to suggest different valid blocks to different
parties. This, however, is easily caught by the atomic broadcast protocol.
Other dishonest parties can report different orders to different leaders (if those exist). This also is caught by the atomic broadcast protocol (which in this case should
select one of those blocks as the next one), as well as requiring contradictory signatures that are then provable exposing the corrupt party.
%be visible in the block validator and thus invalidate the entire block.

\end{proof}

\begin{theorem}
If some honest party submits a request $r$, the protocol {\em Neverending Fairness} terminates.
\end{theorem}

\begin{proof}
Void, because the theorem is wrong.
\end {proof}

As we have shown in the previous section, it is possible for an adversary to construct a schedule in which an arbitrary amount of messages needs to be put into the same block;
thus, an adversary with sufficient influence on message ordering can keep the protocol process one block forever.

Consequently, we also cannot quantify the absolute fairness -- once a request is seen by all honest parties, there is no upper bound on when it is delivered. The only statement we can make is
about the block it will be contained in (which depends on the number of undelivered earlier requests), but not on the time or communication effort until that block 
is delivered.

\begin{lemma}
If $r'$ does not block $r$, then $r'$ is not required to be in the same or an earlier block than $r$ by the requirements of relative block fairness.
\end{lemma} 
\begin{proof}
To be required to be in the same or an earlier block, all honest parties need to have seen $r'$ before $r$. If $t+1$ parties report to have seen
$r'$ after $r$, at least one of them is honest, and thus not all honest parties have seen $r'$ before $r$.
\end{proof}

\subsection{Armageddon}

If the protocol terminates due to lack of usage (i.e., there are no more requests to be scheduled), then the impossibility result no longer holds -- in the worst
case scenario, the protocol only schedules one block after the genesis block which then contains all transactions (one could argue that such a behavior may hasten 
the end-of-time scenario as users abandon the system). What is left to show is that all requests that an honest party has seen actually are delivered. This model also 
assumes that the adversary cannot keep the protocol running forever by generating its own transactions. This would usually be the case as (a) forever is a very long time 
and a concept that doesn't exit in a cryptographically strict model,
(b) usually transactions cost money to incentivise the validators, so such an adversary would spent an unlimited amount of money to prevent protocol termination.

If the protocol terminates while still in operation due to validators opting out, a weaker form of liveness is required -- while the protocol should 
have created all the blocks it could before, it cannot be expected to deliver every single request in that setting. While we do not quantify which messages can 
get lost under these conditions,~\cite{cryptoeprint:2020:269} provides the formalism to cleanly define such end-time scenarios.

\subsection{Relative Synchrony Assumption}

One reason why the impossibility result works is that we allow the adversary to completely control the schedule, i.e., the order in which all parties see all requests. This is an unrealistically strong
adversary; it is usually defined that way as it is rather hard to model a realistic worst case network attack. In the following, we define an adversary who is almost that strong, but has a (small) 
failure probability.
For this definition, we assume that there is some form of global time, which is unknown to the individual parties.

%\begin{definition}[Probabilistic Adversary Failures]
%If two messages $m_1$ and $m_2$ are sent by honest parties such that $m_1$ has been sent before $m_2$, 
%then with at least probability $p$, $m_1$ is delivered before $m_2$.
%\end{definition}

\begin{definition}[Probabilistic Adversary Failures]
After every time the adversary delivers a message, all undelivered messages between honest parties, in a random order,
are each delivered with a probability $p$. If as a result of such a message an honest party generated another message,
that message is added to the pool of messages to be delivered with probability $p$ at a random position.
\end{definition}

While this definition invalidates the impossibility result and allows for an algorithm to achieve relative fairness, we still run into practical issues. If $p$ is unknown (analogous to the failure detectors, 
where it is unknown when a party is rightfully suspected), then we have no known upper bound for the block size and, relatedly, latency. Even if $p$ is known, the maximum possible
blocksize can be prohibitively large for any practical implementation. In addition, the adversary can improve the schedule shown in the previous section to add more resilience. For example, the
adversary could (using twice as many transactions) interweave two such schedules in parallel, and thus tolerate a delivery error in one schedule; to force termination, dilvery errors need to
affect both schedules within a short time, which would then happen only with probability $p^2$. If the adversary has enough messages to operate with, the resilience can thus be arbitrarily high.

While this model is probably pretty close to reality in that a realistic adversry will not have complete control over message delivery for a very long time, it is also unsatisfying in that
$p$ is extremely hard to determine (and probably not the same for all messages and not independent for each message).Furthermore, a more detailed analysis would have to be made
on how an attacker can create even more error- resilient schedules with fewer messages, i.e., how many delivery failures need to co-incide to terminate the protocol.
 Thus, while we can show termination within this model, more work is required to refine the model to the point that we can also make qualitative statements on expected block sizes and latency.

Note that this definition also adds enough synchrony to allow for deterministic byzantine agreement, as the adversary will (eventually) fail to prevent termination.

It remains an open question how much synchrony (in terms of limited message delay) would be needed to circumvent our impossibility result. While we expect that simply
having known timeouts is not sufficient -- our construction only requires requests to be seen in a bad order relative to each other, and also works if all parties see a given
request within a limited time interval -- the exact benefit of various synchrony assumptions are still open work (for some further work on this, see Kelkar,Zhang, Goldfeder, and Juels~\cite{cryptoeprint:2020:269}).
% As argued above, in our construction the adversary
%can handle some amount of synchrony - the number of requests that honest parties need to see between the times the first honest party sees a particular request and the time the last one does
%s limited, and, as the construction only needs to provide the appropriate relative order

\subsection{Probabilistic Relative Block Fairness}
\begin{definition}[Probabilistic Relative Block Fairness]
A byzantine fault tolerant total ordering protocol is called {\em probabilistically relatively block fair} if the following holds: There is a fixed probability $p$ such that,  if all honest parties receive request 
$r_1$  before request $r_2$, then $r_1$ is delivered in the same block as $r_2$ or earlier for with at least probability $p$.
\end{definition}

This definition allows a protocol to at some point stop assuring fairness and put the already processed messages into the next block, 
even if that means that some messages are scheduled unfairly. To achieve termination at sacrificing some level of fairness, we can set a 
threshold $r_{max}$ and artificially terminate the protocol once the number of requests in $R$ exceeds $r_{max}$. This means that an 
adversary with sufficient network control can cause a limited amount of unfairness (i.e. scheduling some requests out of a fair order), 
however, the majority of all requests will be scheduled fairly, and causing an unfair order does require a very high level of network control for the adversary. 
Of course, the cut-off point can also be defined using other factors, e.g., a timeout, the number of requests in the queue, etc.

%\begin{table}
%\centering
%\begin{tabular}{l|lll|llll|l|lll|llll|}
%     & $B_1^1$ &          &          & 
 %       $B_1^2$ &          &          &         &      
 %       $B_2^1$  &
   %     $B_2^2$ &          &          & 
   %     $B_3^1$ &          &          &          
    %    $B_3^2$  \\ 
%\hline
%$P_1$ & $m_1$ & $m_2$ & $m_3$ &   
 %        &          &          &          &
  %     &
   %     & $m_8$ & $m_5$ & $m_6$ & $m_7$&
    %       $m_4$ 

%
%  \\
%$P_2$ & $m_2$ & $m_3$ &          &
 %$m_4$ & $m_1$ &          &          &  
 %        \\
%$P_3$ & $m_3$ &          &          &
% $m_4$ & $m_1$ &          &          &    
 %      \\
%$P_4$ &          &          &          &
% $m_4$ & $m_1$ & $m_2$ & $m_3$ & 
%\\         
%\end{tabular}
%\end{table}

%\end{proof-sketch}

%Unfortunately, even this definition seem a bit strong; as the corrupt parties may try to schedule $r_2$ long before any honest party sees $r_1$,
%it would require the protocol to refuse to schedule any message only seen by $t$ parties (which is an argument that continues, as $r_2$ will %always have been seen bu $t$ more parties than $r_1$. While this is not yet an impossibility proof, it also shows the definition is poor - 
%. 

We can strengthen this approach by adding a random factor. In that setting, once $r_{max}$ is exceeded, we use a common coin~\cite{DBLP:journals/joc/CachinKS05} to determine when the protocol stops. This could
be done in a way that the result is unpredictable even for the leader - after each request added to $B$ beyond $r_{max}$, the leader can request a coin from all other parties defining
whether or not she should stop at that point.  Thus while an adversary with extensive network control can cause an unfair scheduling, she has no influence on who is treated unfairly.
Communication overhead can be managed by piggybacking the coin shares to the voting messages; furthermore, as the attacker gains little apart from a small slowdown of the protocol, 
one could hope that most economic attackers would not attempt such an attack, and thus in most cases the protocol terminates before reaching $r_{max}$.
While this allows the timing model to remain unchanged, the required maximum blocksize is linked to $p$; if $p$ is to be very small (e.g., one in a million), the 
number of messages per block that the protocol needs to be capable of handling is correspondingly high.
 
\subsection{Fairness using Local Clocks}

We now present a different definition of fairness that is slightly weaker, but that allows for much stronger liveness guarantees.

\begin{definition}[Timed Relative Fairness]
Suppose that all parties have access to a local clock. If there is a time $\tau$ such that all honest parties saw (according to their local clock) request $r$ before $\tau$ and request $r'$ after $\tau$, 
then $r$ must be scheduled before $r'$.
\end{definition}

\nocite{DBLP:journals/corr/abs-1904-05234}
Note that there is no need for the local clocks to be synchronized at all; the only formal  requirement is that the clock always counts forward and that no two timestamps are the same.  Obviously, the definition does make more practical sense if the clocks are roughly in sync. 
Using GPS as a time source with a hardening layer to prevent GPS spoofing (e.g., ~\cite{DBLP:journals/ijccbs/BondavalliBCFV13}) and robust syncronisation protocols~\cite{cryptoeprint:2019:838} should be more than sufficient to make this approach practical.

For our protocol, it is sufficient to assure that if $r$ needs to be scheduled before $r'$, $r$ is in an earlier or the same block. As the timestamps are included in the block, 
the ordering of requests inside a block can be performed locally after the block is delivered.

\begin{algorithm}
{\bf Widget Clocked-Wendy for block $b$ and protocol  instance  ID}\\
{\bf All parties:}
\begin{small}
\begin{revpar}
\item {\bf let} $i$ be a counter for incoming requests, starting at 0
\item {\bf while} no valid proposal has been seen as the proposal for atomic broadcast for block $b$  {\bf do}
\begin{revpar}
 \item {\bf for all} known and unscheduled requests $\hat r$, in the order of the timestamps on the requests, send the message   ($ID$,$b$,$i$, timestamp($\hat r$),$\hat r$) to all parties, where $i$ is the sequence number of that request.
% \item {\bf if} a new request $\hat r$ arrives, send the signed message ($ID$,$b$,$i$, timestamp($\hat r$)) to all parties

\end{revpar}
\item {\bf end while}
%\begin{revpar}
%\item {\bf while} no valid proposal has been seen as the proposal for atomic broadcast {\bf do}
%\begin{revpar}
 %\item send all known and unscheduled requests (signed) toall validators in the order in which you saw them
 %\item {\bf if} a new request arrives, send the new request to all validators, with a timestamp and a signature on the complete list.
% \end{revpar}
%\item {\bf end while}

\end{revpar}
\end{small}
{\bf Additional protocol for the leader(s):}
\begin{small}
\begin{revpar}
\item {$\mathcal{B} = \emptyset$}
\item {\bf wait until } the first requests $r$ is contained in the signed list of $n-t$ validators; add $r$ to $\mathcal{B}$
\item {\bf let}  $\mathcal{R}$ be the set of requests for which a vote for with a smaller timestamp than $r$ was received
\begin{revpar}
 \item 	{\bf wait until} there is a set of $n-t$ parties from which valid votes for all requests in $\mathcal{R}$ are received
\end{revpar}
 {\bf for all} $r' \in \mathcal{R}$, if timestamps of $t+1$ votes are smaller for $r'$ than the median of the timestamp of the votes for $r$, add $r'$ to $\mathcal{B}$

\item The proposal for the next block of the atomic broadcast is $\mathcal{B}$, validated be the corresponding signed votes in $\mathcal B$
\end{revpar}
\end{small}
\end{algorithm}

%XXX NOTE: I can wait until there is a set $n-t$ for ALL parties that had $t+1$ votes
Since the fairness condition changed, the validity of a vote and of a block also look different.

\begin{definition}[Timestamped Vote-Validity]
A vote is valid if it has the proper format, and if the sequence number matches the sequence on timestamps on requests from that party. Once a party mismatches 
the timestamps and the sequence numbers, i.e., there are two requests $r_1$ and $r_2$ such that $r_1$ has a lower sequence number and a higher timestamp than $r_2$, 
this and all following votes from that party are considered invalid. Furthermore, a vote is only considered valid once all requests with a lower sequence number from that voter have been received.
\end{definition}

\begin{definition}[Timestamped Block-Validity]
A block $B$ is valid if it contains a nonempty  set of requests with $n-t$ valid votes each; a vote for $r$ is valid if it contains the signed votes for all requests for that block with a lower sequence number. Furthermore, for every $r$ in $\mathcal{B}$, if there is a request $r'$ in the vote validation that obtained $t+1$ votes with a lower sequence number than $r$, then $r'$ needs to be in $\mathcal{B}$ accompanied 
by $n-t$ validation votes.
\end{definition}

\begin{theorem} (Safety)
If a request $r$ is scheduled in a block $\mathcal{B}$, and there is a request $r'$such that there is a time $\tau$ in a way that all honest parties saw $r'$ before $\tau$ and $r$ after $\tau$, then 
$r'$ is in $\mathcal{B}$ or an earlier block.
\end{theorem}
\begin{proof}

Assume without loss of generality that every timestamp has a unique time. This can easily be assured locally by a high 
enough time resolution, and by ordering votes by party identifier if two votes have the exact same timestamp.

Suppose at the end of the pre-protocol, we have request $r' \in \mathcal{B}$ and $l \not \in \mathcal{B}$, and that $l$ has not been scheduled in an earlier block. Let $\tau_1$ be the median of the
timestamps of $r$. 

\begin{enumerate}
\item As $r' \in \mathcal{B}$, at least $t+1$ parties timestamped $r'$ before or during $\tau_1$
\item As $l \not \in \mathcal{B}$, at most $t$ parties timestamped $l$ before $\tau_1$.
\end{enumerate}
 Suppose by the requirements of timed relative fairness, we have to schedule $l$ before $r'$. As $t+1$ of the parties that issued votes are honest, this implies
that 
\begin{enumerate}
\setcounter{enumi}{2}
\item there exists $\tau_2$ such that $t+1$ votes contain timestamps for $l$ before $\tau_2$, and at most $t$ votes contain timestamps for $r'$ before $\tau_2$.
\end{enumerate}

By (2), at most $t$ timestamps for $l$ are smaller than $\tau_1$, and by (3) at least $t+1$ are smaller than $\tau_2$; thus, $\tau_1$ is smaller than $\tau_2$.
Similarly, for $r'$, by (3) at most $t$ timestamps are smaller than $\tau_2$, by (1)  and at least $t+1$ are smaller or equal to $\tau_1$. Thus, $\tau_2$ is smaller than $\tau_1$.
This is a contradiction, and therefore it is not possible that $l$ needs to be scheduled before $r'$. 
\end{proof}

\begin{theorem} 
 If some honest party sees some request, any honest leader will terminate the protocol with a proposal.
\end{theorem}

\begin{proof}
As every party sends every request it sees for the first time to all other parties,  every request that is seen by some honest party is
seen -- and send to the leader(s) -- by all honest parties. Thus, there is some $r$ that is in the signed list of $n-t$ parties.
Once a leader gets $n-t$ votes for some $r$ for the first time, there is a finite number of requests $r'$ for which the leader received a vote before. 
As the leader has seen this vote and is honest, it also forwarded the $r'$ to all other parties, and thus will receive $n-t$ votes eventually.  
Therefore, the waiting statement always terminates for all requests $r'$. 
\end{proof}

Note: We only need successful termination if an honest leader exists. All atomic broadcast protocols we are aware of either have a single leader which 
is replaced if a liveness problem occurs, or use more than $t$ parties in a leader-like function simultaneously and thus guarantee that there is some
honest leader.

\subsection{Optimizations}

The two protocols described above can also be combined. The joint protocol would act like the neverending protocol up until $r_{max}$; however, instead of aborting the
protocol and allowing for plain unfairness, it switches to the weaker timed definition of fairness once $r_{max}$ is exceeded. That approach allows for much more
aggressive thresholds, as the fallback protocol is no longer unfair, but still fair with a slightly weaker definition.

\subsubsection{Latency and performance impact}

Introducing any kind of relative fairness always has a latency impact. If no fairness is required, every incoming request can be processed as soon as it arrives. Relative
fairness, no matter how it is defined,  requires leaders to wait until they can decide if there are other requests with a higher priority. While the {\em mostly fair} protocol allows  parameterisation of  the trade-off between latency and unfairness -- the lower the cutoff parameter, the faster the worst case  protocol and the easier for an adversary to cause an unfair transaction. However, in the benign
case, the latency overhead should be reasonably small.

One (small) speed increase can be reached by parallelizing the leader part of the protocols. Instead of waiting for the first request to add to $\mathcal B$ and then sticking to it,
the protocol can be run in parallel for all requests that have been reported by enough parties. In that case, the first instance that terminates its {\bf while} condition wins
and defines the next block. It is also possible to cut the threshold in the neverending fairness protocol to $t+1$ by using a more sophisticated blocking function.

Another parallelization approach would be that the first part of the protocol  where all parties broadcast their orders is permanently performed, independently of the state of the second
phase or the atomic broadcast. Thus, in most cases, once the atomic broadcast starts processing the next block, enough votes should have arrived to terminate the pre-protocol
quite rapidly. This approach also has an interesting impact on the overall architecture -- rather than having a simple API to call the pre-protocol, some part of it now needs to permanently
run in the background. Alternatively, to save overhead, this could also be included as a piggyback in the gossiping protocol.

An additional approach to optimize the Neverending protocol is to allow requests to be removed from $\mathcal B$ again. Recall that a request is added to $\mathcal B$ if it has received $n-t$ votes and still blocks a request already in $\mathcal B$. This is necessary as we can no longer rely on getting more votes concerning this request, and to guarantee progress, this
request now needs to be treated as if we know that it has to be in the same block as the one it blocks. However, as additional votes come in, it is possible that it unblocks again. In this case, 
$r$ and all requests that had been added to $\mathcal B$ due to blocking $r$ can be removed from $\mathcal B$ again, potentially releasing the block earlier.

For the timed protocol, a similar approach can be taken. For this protocol, we have the advantage that for each request $r$, there is a finite number of requests that are blocking it. This
blockage is released either once the corresponding request has $t+1$ timestamps smaller than the median timestamp on $r$ (in which case we know if any other request needs to be
scheduled before $r'$, it also needs to be scheduled before $r$), or if it got $n-t$ timestamps of which at most $t$ are smaller than the median of $r$ (in which case it can and will be scheduled
after $r$). To fully optimize latency, we also need to constantly verify if new incoming votes increase the median of a subset of $n-t$ votes for $r$, as a higher median increases the possibility
that another request can be decided before it obtained $n-t$ votes.

With this modification, we believe that the protocols have optimal latency within our modular architecture, i.e., it is not possible to hand a block over to the atomic broadcast protocol earlier.
 The (informal) argument for the block fairness protocol goes as follows (from the point of view of a leader): 
\begin{itemize}

\item Every request that $r$ got $t+1$ votes gets its own $\mathcal {B}_r$, i.e., a potential block  containing $r$ and all other requests that have to be in the same block as $r$ .By our
fairness condition, we cannot deliver any request that has seen less than $t+1$ votes, as it is possible that another request that is unknown at this point has $n-t$ votes that prioritize it
over $r$ and thus has to be in the same block. Therefore, for every request that can be in the next block, the protocol maintains has a ${\mathcal B}_r$
\item At any point in time, ${\mathcal B}_r$ is minimal; the only requests in ${\mathcal B}_r$ are requests that either have to be in the same block as $r$, or might have to according to the information
        available.
\item   ${\mathcal B}_r$ cannot be finalized while it contains a request $r_1$ that is blocked by another request $r_2$ with less than $t+1$ votes, as $r_2$ might still be blocked by a yet unseen request. 
       Thus, the protocol finalizes ${\mathcal B}_r$ at the earliest possible occasion.
\end{itemize}

A similar argument holds for the timed protocol; again, the protocol maintains a separate ${\mathcal B}'_r$ for all eligible processes, and decides about all other requests at the earliest
opportunity -- either once it is clear that they can to be in the same block, or once enough votes are seen to conclude they don't need to.

If we further want to optimize latency, we could open up the modularity of our approach. Most voting based atomic broadcast protocols start with the leader(s) broadcasting the content of the next block (or a hash thereof). 
Due to the pre-protocol, we already know that $n-t$ parties have seen the content of the requests in that block.  Optimizing the interplay between the fairness pre-protocol, the atomic broadcast, and the underlying gossip/multicast protocol is thus certainly promising, but out of the scope of this paper. It also is possible to integrate our protocol deeper with the blockchain implementation. With some modifications it could, for example, replace the first phase of the ABC protocol from Cachin, Kursawe, Petzold and Shoup~\cite{CKPS01}. As our goal is a modular approach though, we will not follow that path at this point.

\subsubsection{The combined protocol}

There is a set $\mathcal{D}$ of transactions that are ready for the atomic broadcast layer to use. For the ease of presentation, we
assume that the communication layer is aware of $\mathcal D$, and omits any voting messages associated with any transaction in 
%XXX DOES THAT WORK IF WE SWITCH LEADERS ??
$\mathcal D$. Furthermore, there is a queue $\mathcal Q$ with which the protocol communicates with the atomic broadcast. The atomic broadcast
protocol takes the requests in $\mathcal Q$ from one or several leaders, adds a block to the blockchain, and then deletes the scheduled
requests from the queues from all leaders.

This version of the protocol is defined as a permanent service that takes in requests, and outputs blocks for the atomic broadcast protocol.
%We also assume the atomic broadcast can send a signal {\em abcast} once it picked up a non-empty set of transaction to process. This will be accessed as the boolean variable in the protocol.
 
%XXX no longer for block b

\begin{algorithm}\label{prot:hybridfairness}
{\bf Widget Hybrid-Wendy for protocol instance $ID$}\\
{\bf All parties:}
\begin{small}
\begin{revpar}
\item {\bf let} $i$ be the counter of incoming requests, starting at 0.
\item {\bf while} true  {\bf do}
\begin{revpar}
 \item {\bf for all} first seen and unscheduled requests $\hat r$, in the order of the timestamps on the requests, send the message   ($ID$,$b$,$i$, timestamp($\hat r$),$\hat r$) to all parties, where $i$ is the sequence number of that request.
 %\item {\bf if} a new request $\hat r$ arrives, send the signed message ($ID$,$b$,$i$,$\hat r$) to all parties
\end{revpar}
\item {\bf end while}

\end{revpar}
\end{small}
{\bf Additional protocol for the leader(s):}
\begin{small}
\begin{revpar}

%\item {$\mathcal{D} = \emptyset$}
\item {\bf while} true  {\bf do}
\begin{revpar}
%\item remove all votes relating to a request in $\mathcal D$ and ignore any such incomming votes.
\item {\bf once } a request $r$ is contained in the signed and valid votes from $t+1$ parties, set  $\mathcal{B}_r$ to $\{r\}$
%XXX CHeck if t+1 is ok here
\item {\bf while} for any $\mathcal{B}_r \neq \emptyset$  any request $r' \not \in \mathcal{B}_r$ blocks a request $r \in \mathcal{B}_r$ {\bf and} there is no $\mathcal{B}_x$ of order $> r_{max}$
\begin{revpar}
\item 	{\bf if} request $r'$  has at least $t+1$ votes, add $r'$ to $\mathcal{B}_r$
\item      {\bf if} a request $r'  \in \mathcal{B}_r, r' \neq r, $ no longer blocks any other request in $\mathcal{B}_r$, remove $r'$ from $\mathcal{B}_r$
\end{revpar}
\item {\bf end while}

\item {\bf for all} $r$ for which no request in  $\mathcal{B}_r$ is blocked by a request $r' \not \in \mathcal{B}_r$, 
\begin{revpar}
\item add  $\mathcal{B}_r$ to the $\mathcal{Q}$, validated by all signed votes for requests in $\mathcal{B}_r$.
\item add all $r' \in \mathcal{B}_r$ to $\mathcal{D}$, and remove them from all sets $\mathcal{B}_x$ 
\end{revpar}

\item {\bf if} there is a $\mathcal{B}_x$ of order $> r_{max}$
\begin{revpar}
%%$ CUT FROM HERE...
\item set all $\mathcal{B}_x = \emptyset$
\item {\bf while} all $\mathcal {B}'_x = \emptyset$
\begin{revpar}
\item {\bf for all } requests $r$  contained in the signed list of $n-t$ validators, set  $\mathcal{B}'_r = \{r\}$
\end{revpar}
%XXX ... to here

\item {\bf for all} $r$ relating to a nonempty $\mathcal{B}'_r$, 
\begin{revpar}
\item {\bf let}  $\mathcal{R}_r$ be the set of requests for which a vote with a smaller timestamp than $r$ was received
\item {\bf let} $m_r$ be the largest median of any set of $n-t$ votes received for $r$
%\item {\bf once} there is a set of $n-t$ parties from which valid votes for all requests in some  $\mathcal{R}_r$ are received
\item {\bf once} for all requests in ${\mathcal R}_r$ $n-t$ valid votes {\bf or} $t+1$ votes with timestamps smaller than $m_r$ are received, 
\begin{revpar}
\item {\bf for all} $r' \in \mathcal{R}_r$, if timestamps of $t+1$ votes are smaller for $r'$ than $m_r$, add $r'$ to $\mathcal{B}_r$
\end{revpar}
\end{revpar}
\item add  $\mathcal{B}'_r$ to $\mathcal{Q}$, validated by all signed votes for requests in $\mathcal{B}'_r$.
\item add all $r' \in \mathcal{B}'_r$ to $\mathcal{D}$, and remove them from all $\mathcal{B}'_x$

\end{revpar}
\item {\bf end if}

\end{revpar}
\item {\bf end while}
\end{revpar}

\end{small}
\end{algorithm}

\subsection{Fairness and Advanced Staking}
While the protocol described above is relatively model-independent, it is described in the classical committee model, i.e., we have $n$ parties with one vote each, up to $t, n \geq 3t+1$ can suffer from byzantine corruptions.
This model translates easily into a stake-based model, where voting power is related to the stake parties have. To allow our results to be applicable for more different staking models,
we consider the hybrid-adversary-structure model~\cite{Kursawe05byzantinefault}. In short, this model generalizes the model by replacing the thresholds by the corresponding properties that are required
to perform the proof; for example, the threshold $t+1$ is replaced by {\em sets of parties of which at least one is honest}, while $n-t$ corresponds to {\em the largest sets of parties we can afford
to wait for without having to rely on potentially corrupt parties}. This allows to not only model weighted votes, but also take into account properties, e.g., requiring more $2/3$ of the stake in more than $2/3$ of a set of defined geographic regions to be honest. In addition, the hybrid model allows a trade-off between crash- and byzantine corruptions, allowing a higher number of overall failures if some of them are
crash-only (which is a more likely scenario in reality). In the proofs for our protocols, the two aforementioned properties are the only properties we need, and -- while working out the details remains future work -- we expect that the proofs can be generalized in a (relatively) straightforward way. Thus, any staking model that can be formulated within the hybrid adversary structures is compatible with the relative-fairness protocols.

Another model of interest is the choice of a random subset of validators, as done for example in Algorand~\cite{10.1145/3132747.3132757}. While we do not expect this to cause a fundamental issue for our protocols,
some care needs to be taken on the interfaces, as the subsets should be somewhat synchronised between the fairness pre-protocol and the atomic broadcast. This, too, will be a subject of future work.

%\input{section.generalfairness.tex}
%\include{section.antispam}

%\cite{CKPS01}
%\begin{figure}    \centering
 %   \incl#degraphics[width=3.0in]{myfigure}
 %   \caption{Simulation Results}
 %   \label{simulationfigure}
%\end{figure}

\section{Conclusion}
We have shown that relative fairness is one of the many desirable properties that is impossible to achieve in a byzantine fault tolerant setting. We have mitigated this by providing slightly weaker definitions of what fair is. We have a presented several protocols to achieve relative order fairness with these definitions, as well as a hybrid version that can switch between two levels of fairness to avoid the impossibility result. 
Our protocols are (largely) blockchain agnostic, and can be added to any protocol that provides a known set of validators. Furthermore, our protocols have optimal resiliency in the asynchronous model (i.e., $n \geq 3t+1$) and
optimal latency in terms of message passing rounds within our architectural model.
\bibliography{wendy}

\begin{thebibliography}{10}

\bibitem{DBLP:journals/corr/abs-1803-05069}
I.~Abraham, G.~Gueta, and D.~Malkhi.
\newblock Hot-stuff the linear, optimal-resilience, one-message {BFT} devil.
\newblock {\em CoRR}, abs/1803.05069, 2018.

\bibitem{DBLP:conf/netys/Amoussou-Guenou19}
Y.~Amoussou{-}Guenou, A.~D. Pozzo, M.~Potop{-}Butucaru, and S.~{Tucci
  Piergiovanni}.
\newblock Dissecting tendermint.
\newblock In M.~F. Atig and A.~A. Schwarzmann, editors, {\em Networked Systems
  - 7th International Conference, {NETYS} 2019, Marrakech, Morocco, June 19-21,
  2019, Revised Selected Papers}, volume 11704 of {\em Lecture Notes in
  Computer Science}, pages 166--182. Springer, 2019.

\bibitem{cryptoeprint:2019:838}
C.~Badertscher, P.~Gazzi, A.~Kiayias, A.~Russell, and V.~Zikas.
\newblock Ouroboros chronos: Permissionless clock synchronization via
  proof-of-stake.
\newblock Cryptology ePrint Archive, Report 2019/838, 2019.
\newblock https://eprint.iacr.org/2019/838.

\bibitem{DBLP:journals/ijccbs/BondavalliBCFV13}
A.~Bondavalli, F.~Brancati, A.~Ceccarelli, L.~Falai, and M.~Vadursi.
\newblock Resilient estimation of synchronisation uncertainty through software
  clocks.
\newblock {\em {IJCCBS}}, 4(4):301--322, 2013.

\bibitem{DBLP:journals/corr/abs-1807-04938}
E.~Buchman, J.~Kwon, and Z.~Milosevic.
\newblock The latest gossip on {BFT} consensus.
\newblock {\em CoRR}, abs/1807.04938, 2018.

\bibitem{DBLP:journals/corr/abs-1710-09437}
V.~Buterin and V.~Griffith.
\newblock Casper the friendly finality gadget.
\newblock {\em CoRR}, abs/1710.09437, 2017.

\bibitem{CKPS01}
C.~Cachin, K.~Kursawe, F.~Petzold, and V.~Shoup.
\newblock Secure and efficient asynchronous broadcast protocols.
\newblock In J.~Kilian, editor, {\em Advances in Cryptology - {CRYPTO} 2001,
  21st Annual International Cryptology Conference, Santa Barbara, California,
  USA, August 19-23, 2001, Proceedings}, volume 2139 of {\em Lecture Notes in
  Computer Science}, pages 524--541. Springer, 2001.

\bibitem{DBLP:journals/joc/CachinKS05}
C.~Cachin, K.~Kursawe, and V.~Shoup.
\newblock Random oracles in constantinople: Practical asynchronous byzantine
  agreement using cryptography.
\newblock {\em J. Cryptology}, 18(3):219--246, 2005.

\bibitem{DBLP:conf/osdi/CastroL99}
M.~Castro and B.~Liskov.
\newblock Practical byzantine fault tolerance.
\newblock In M.~I. Seltzer and P.~J. Leach, editors, {\em Proceedings of the
  Third {USENIX} Symposium on Operating Systems Design and Implementation
  (OSDI), New Orleans, Louisiana, USA, February 22-25, 1999}, pages 173--186.
  {USENIX} Association, 1999.

\bibitem{DBLP:journals/corr/abs-1904-05234}
P.~Daian, S.~Goldfeder, T.~Kell, Y.~Li, X.~Zhao, I.~Bentov, L.~Breidenbach, and
  A.~Juels.
\newblock Flash boys 2.0: Frontrunning, transaction reordering, and consensus
  instability in decentralized exchanges.
\newblock {\em CoRR}, abs/1904.05234, 2019.

\bibitem{VegaWhitepaper}
G.~Danezis, D.~Hrycyszyn, B.~Mannering, T.~Rudolph, and D.~\v{S}i\v{s}ka.
\newblock Vega protocol: A liquidity incentivising trading protocol for smart
  financial products.
\newblock 2018.

\bibitem{8023111}
S.~{Duan}, M.~K. {Reiter}, and H.~{Zhang}.
\newblock Secure causal atomic broadcast, revisited.
\newblock In {\em 2017 47th Annual IEEE/IFIP International Conference on
  Dependable Systems and Networks (DSN)}, pages 61--72, 2017.

\bibitem{10.1145/42282.42283}
C.~Dwork, N.~Lynch, and L.~Stockmeyer.
\newblock Consensus in the presence of partial synchrony.
\newblock {\em J. ACM}, 35(2):288–323, Apr. 1988.

\bibitem{10.1145/3132747.3132757}
Y.~Gilad, R.~Hemo, S.~Micali, G.~Vlachos, and N.~Zeldovich.
\newblock Algorand: Scaling byzantine agreements for cryptocurrencies.
\newblock In {\em Proceedings of the 26th Symposium on Operating Systems
  Principles}, SOSP ’17, page 51–68, New York, NY, USA, 2017. Association
  for Computing Machinery.

\bibitem{cryptoeprint:2020:269}
M.~Kelkar, F.~Zhang, S.~Goldfeder, and A.~Juels.
\newblock Order-fairness for byzantine consensus.
\newblock Cryptology ePrint Archive, Report 2020/269, 2020.
\newblock \url{https://eprint.iacr.org/2020/269}.

\bibitem{Kursawe05byzantinefault}
K.~Kursawe and F.~C. Freiling.
\newblock Byzantine fault tolerance on general hybrid adversary structures.
\newblock Technical report, RWTH Aachen, 2005.

\bibitem{DBLP:conf/icalp/KursaweS05}
K.~Kursawe and V.~Shoup.
\newblock Optimistic asynchronous atomic broadcast.
\newblock In L.~Caires, G.~F. Italiano, L.~Monteiro, C.~Palamidessi, and
  M.~Yung, editors, {\em Automata, Languages and Programming, 32nd
  International Colloquium, {ICALP} 2005, Lisbon, Portugal, July 11-15, 2005,
  Proceedings}, volume 3580 of {\em Lecture Notes in Computer Science}, pages
  204--215. Springer, 2005.

\bibitem{Kwiatkowska02verifyingrandomized}
M.~Kwiatkowska and G.~Norman.
\newblock Verifying randomized byzantine agreement.
\newblock In {\em Proc. Formal Techniques for Networked and Distributed Systems
  (FORTE’02), volume 2529 of LNCS}, pages 194--209. Springer, 2002.

\bibitem{DBLP:journals/corr/abs-1906-03819}
K.~Lev{-}Ari, A.~Spiegelman, I.~Keidar, and D.~Malkhi.
\newblock Fairledger: {A} fair blockchain protocol for financial institutions.
\newblock {\em CoRR}, abs/1906.03819, 2019.

\bibitem{Miller:2016:HBB:2976749.2978399}
A.~Miller, Y.~Xia, K.~Croman, E.~Shi, and D.~Song.
\newblock The honey badger of bft protocols.
\newblock In {\em Proceedings of the 2016 ACM SIGSAC Conference on Computer and
  Communications Security}, CCS '16, pages 31--42, New York, NY, USA, 2016.
  ACM.

\bibitem{10.1145/177492.177745}
M.~K. Reiter and K.~P. Birman.
\newblock How to securely replicate services.
\newblock {\em ACM Trans. Program. Lang. Syst.}, 16(3):986–1009, May 1994.

\end{thebibliography}
%\bibliography{BIB/model,BIB/crypto,BIB/byz,BIB/threshold,BIB/misc,BIB/mycrypto,BIB/dist,BIB/proactive,BIB/abba,BIB/proceedings,BIB/THESIS,BIB/abc,BIB/abc1aux,BIB/abc2aux}
\bibliographystyle{abbrv}

\end{document}